\documentclass[transmag]{IEEEtran}
\usepackage{command_MU}
\usepackage{latexsym}
\usepackage{graphicx}
\usepackage{amsfonts,amssymb,amsmath}
\usepackage{hyperref}
\usepackage{booktabs}
\usepackage{colortbl}
\usepackage{algorithm}
\usepackage{algpseudocode}
\usepackage{lscape}
\usepackage{subfig}
\usepackage{graphicx}
\usepackage{mathtools}
\usepackage{booktabs}
\DeclareMathOperator*{\argmin}{arg\,min}

\def\BibTeX{{\rm B\kern-.05em{\sc i\kern-.025em b}\kern-.08em T\kern-.1667em\lower.7ex\hbox{E}\kern-.125emX}}


\begin{document}

\title{Delaunay-Triangulation-Based Learning with Hessian Total-Variation Regularization}

\author{Mehrsa Pourya, Alexis Goujon, and Michael Unser,
\IEEEmembership{Fellow, IEEE}
\thanks{This work was supported in part by the European Research Council (ERC Project FunLearn) under Grant 101020573 and in part by the Swiss National Science Foundation, Grant 200020\_184646/1.}
\thanks{The authors are with the Biomedical Imaging Group, \'Ecole polytechnique f\'ed\'erale de Lausanne, 1015 Lausanne, Switzerland (e-mail: mehrsa.pourya@epfl.ch; alexis.goujon@epfl.ch; michael.unser@epfl.ch)}
}

\IEEEtitleabstractindextext{\begin{abstract}
Regression is one of the core problems tackled in supervised learning. Rectified linear unit (ReLU) neural networks generate continuous and piecewise-linear (CPWL) mappings and are the state-of-the-art approach for solving regression problems. In this paper, we propose an alternative method that leverages the expressivity of CPWL functions. In contrast to deep neural networks, our CPWL parameterization guarantees stability and is interpretable. Our approach relies on the partitioning of the domain of the CPWL function by a Delaunay triangulation. The function values at the vertices of the triangulation are our learnable parameters and identify the CPWL function uniquely. Formulating the learning scheme as a variational problem, we use the Hessian total variation (HTV) as regularizer to favor CPWL functions with few affine pieces. In this way, we control the complexity of our model through a single hyperparameter. By developing a computational framework to compute the HTV of any CPWL function parameterized by a triangulation, we discretize the learning problem as the generalized least absolute shrinkage and selection operator (LASSO). Our experiments validate the usage of our method in low-dimensional scenarios. 

\end{abstract}


}

\maketitle

\section{Introduction}

\IEEEPARstart{S}{upervised} learning entails finding a function $f: \R^d \rightarrow \R $ using a set of $M$ training data points $\{\V x_m\}_{m=1}^{M} \subset \R^d$ and their target values $\{y_m\}_{m=1}^{M} \subset \R$. The function $f$ should approximate the target values at the data points $y_m \approx f(\V x_m)$ and generalize well to new inputs \cite{hastie2009overview}. In a variational framework, the learning problem is formalized as the optimization task
\begin{equation}
  \min\limits_{f \in \mathcal{X}} \sum_{m=1}^M \Op{E}(f(\V x_m), y_m) + \lambda \mathcal{R}(f),
\label{variational_learning}
\end{equation}
where $\mathcal{X}$ is the function search space, $\Op{E}: \R \times \R \rightarrow \R_{+}$ is a convex loss function controlling the data-fitting error, and $\mathcal{R}: \mathcal{X} \rightarrow \R$ is the regularizer. The hyperparameter $\lambda \geq 0$ adjusts the contribution of the regularizer. Regularization is used to promote functions with desirable structures such as sparsity and to reduce overfitting \cite{tian2022comprehensive, scholkopf2002learning}. 

In order to make (1) tractable, the search space $\mathcal{X}$ is usually expressed as a parametric space. For instance, linear regression---the simplest model---reduces the learning problem to the search for $\V a \in \mathbb{R}^d$ and $b \in \mathbb{R} $ such that $f(\V x) = \V a^\intercal \V x + b$ \cite{gross2003linear}.  This model is very well understood, but it is severely lacking expressivity. Another approach to the learning problem is founded on the theory of reproducing-kernel Hilbert spaces (RKHS) \cite{kadri2010nonlinear}, \cite{hofmann2005tutorial}. If the search space in (1) is the reproducing-kernel Hilbert space $\mathcal{X} = \mathcal{H}(\R^d)$ with kernel ${\rm k}: \R^d\times \R^d \to \R$, and $\mathcal{R}(f)=\norm{f}^2_{\mathcal{H}}$ where $\norm{\cdot}_{\mathcal{H}}$  is the Hilbert-space norm, then the RKHS representer theorem states that the solution of (1) admits the closed-form
\begin{equation}
    f(\cdot) = \sum_{m=1}^M a_m {\rm k}(\cdot, \V x_m)
\end{equation}
 for some coefficients $(a_m) \in \R^M$ \cite{scholkopf2001generalized}. There, the problem is recast into a solvable linear model. Moreover, by using radial-basis functions as kernels, one can approach continuous functions as closely as desired \cite{schaback2007practical}, \cite{miller1992review}, \cite{micchelli2006universal}. For high-dimensional data, however, kernel methods are outperformed by deep neural networks, which are the state-of-the-art in many applications \cite{abiodun2018state}. A neural network of depth $L$ forms a mapping
\begin{equation}
    \V x \mapsto (\V f_{\V \theta_L} \circ \V \sigma_{L-1} \circ \cdots \circ \V \sigma_l \circ \V f_{\V \theta_l} \cdots \circ \V \sigma_1 \circ \V f_{\V \theta_1}) (\V x),
\end{equation}
where $\V f_{\V \theta_l} : \mathbb{R}^{d_{l-1}} \rightarrow \mathbb{R}^{d_l}$ is a learnable affine mapping parameterized by $\V \theta_l$ and $\V \sigma_{l}$ is the nonlinearity (a.k.a. ``activation function") performed in the $l$th layer. Among all possible activation functions, the most common choice and usually the best in terms of performance is the rectified linear unit ReLU($x$) = $\max(x, 0)$ \cite{schmidt2020nonparametric}. It is known that a ReLU network produces continuous and piecewise-linear (CPWL) functions \cite{montufar2014number}. Conversely, any CPWL function can be described by a ReLU network \cite{arora2016understanding}. The CPWL model is universal, in the sense that it can approximate any continuous mapping \cite{DBLP:conf/iclr/ParkYLS21}. However, a major drawback of letting deep networks learn a CPWL function is the lack of their interpretability, in the sense that the effect of each parameter on the generated mapping is not understood explicitly \cite{fan2021interpretability}.  

There also exists a variational interpretation of shallow ReLU networks \cite{parhi2021kinds, unser2022ridges, ongie2019function}. However, shallow networks cannot represent all CPWL functions, especially in high dimensions. In fact, no variational interpretation is known for deep networks, to the best of our knowledge. Interestingly, there exist other variational problems that admit global CPWL minimizers \cite{debarre2019b}, \cite{debarre2021sparsest}. If we investigate (1) in the one-dimensional case $f: \R \to \R$, with the second-order total variation ${\rm TV^{(2)}}(f)$ as the regularizer, and if we restrict the search space to functions with a bounded second-order total variation, then the extreme points of the solution set are necessarily of the form
\begin{equation}
    f: x \mapsto ax + b + \sum_{k=1}^{K} d_k (x - \tau_k)_+,
\end{equation}
where $a, b \in \R, \V d = (d_k) \in \R^K$, $K < M$, and $(\tau_k) \in \R^K$. \cite{unser2021unifying} .  For such solutions, the regularization has the simple closed-form ${\rm TV^{(2)}}(f) = \norm{\V d}_{\ell_1}$.  As the $\ell_1$-norm promotes sparsity, a regularization by the second-order total variation will promote CPWL functions with few knots. A multidimensional generalization of ${\rm TV^{(2)}}$  is the Hessian total-variation ($\HTV$) seminorm \cite{aziznejad2021measuring}. $\HTV$ is equivalent to ${\rm TV^{(2)}}$ in the one-dimensional case and a similar regularizer was introduced before for image-reconstruction purposes \cite{lefkimmiatis2011hessian}. Since the null space of $\HTV$ is composed of affine mappings, this regularization favors solutions that are predominantly affine; typically, CPWL functions with few pieces. Authors in \cite{campos2021learning} use $\HTV$ regularization to learn two-dimensional functions. They parameterize the function search space by a span of linear box splines which are themselves CPWL functions \cite{condat2006three}. Although their parameterization has a good approximation power, it has some drawbacks. First, the number of box-spline basis functions grows exponentially with the dimension. Second, there are many domain regions without corresponding data points. These regions add complexity to the model without benefiting to its performance. 

In this paper, we investigate an alternative approach that is based on the Delaunay triangulation in an irregular setting. Here are our primary contributions.
\begin{enumerate}
  \item Flexible parameterization of CPWL functions. We partition the domain of CPWL functions by the Delaunay triangulation on a set of adaptive grid points; the function values at the grid points are our learnable parameters and identify the CPWL functions uniquely.
  \item Explicit computation of $\HTV$ in any dimension. For the proposed CPWL model, we show that $\HTV$ is the $\ell_1$-norm applied to a linear transformation of the grid-point values.
  \item Experimental validation of the scheme in dimensions at least two, with real data.
\end{enumerate}

Our learning approach uses the expressivity of CPWL functions in a stable and interpretable manner \cite{goujon2022stable}. These properties stem from the nature of the proposed Delaunay-based parameterization. In addition, the use of $\HTV$ regularization enables us to control the complexity of the final mapping by a single hyperparameter. 

Our paper is organized as follows: In Section II, we introduce the mathematical tools that we need to develop our method. We define Delaunay triangulation, CPWL functions, and the Hessian total variation. In Section III, we describe our CPWL parameterization, explain the procedure to calculate its $\HTV$, and derive the generalized least-absolute-shrinkage-and-selection-operator (LASSO) formulation of our learning problem. Finally, we present our experimental results in Section IV.

\section{Preliminaries}

\subsection{Delaunay Triangulation}
In the $d$-dimensional space $\R^d$, the convex hull of $d+1$ affinely independent points forms a polytope known as a $d$-simplex, or simplex for short. These simplices are indeed triangles and tetrahedrons in 2- and 3-dimensional spaces. A triangulation of a set $\M X \subset \R^d$ of points is the partition of their convex hull into simplices such that any two simplices intersect either at a face joint or not at all. Also, the triangulation vertices, referred to as grid points, are exactly the points themselves. We consider two simplices as neighbors if their intersection is a facet, which is a ($d-1$)-simplex. In general, a triangulation of a set of points is not unique. A celebrated triangulation method is the Delaunay triangulation. 
\begin{definition}[Delaunay Triangulation]
For a set $\M X$ of points in $\R^d$, a Delaunay triangulation is the triangulation $DT(\M X)$ such that no point in $\M X$ is inside the circum-hypersphere of any simplex in $DT(\M X)$. 
\end{definition}

Simply put, the Delaunay triangulation partitions the convex hull of points into well-formed simplices. Specifically for $d=2$, it is known that the Delaunay triangulation maximizes the minimal angle of the triangles of the triangulation and avoids skinny triangles. Similar optimal properties exist in higher dimensions \cite{rajan1994optimality}. In addition, there exist computational methods that produce Delaunay triangulations in any dimension \cite{lee1980two}, \cite{cignoni1998dewall}. 
\subsection{Continuous and Piecewise-Linear Functions}
\begin{definition}[CPWL function]
A function $f: \mathbb{R}^d \rightarrow \mathbb{R}$ is continuous and piecewise-linear if
\begin{itemize}
    \addtolength{\itemindent}{0.4cm}
    \item[1.] it is continuous;
    \item[2.] its domain $\Omega = \bigcup_{n} P_n$ can be partitioned into a set of non-overlapping polytopes $P_n$ over which it is affine, with $\left.f\right|_{P_n}(\V x) = \V a_n^T \V x + b_n$.
\end{itemize}
\end{definition}
 
 The gradient of the function over each polytope or, equivalently, each linear region $P_n$, is $\V \nabla f|_{P_n}(\V x) = \V a_n$. We denote the intersection facet of two neighboring polytopes $P_n$ and $P_k$ as $L_{n,k}$. The ($d-1$)-dimensional volume of the intersection is denoted by ${\rm Vol}_{d-1}(L_{n,k})$. For $d=2$ and $d=3$, this volume corresponds to length and area, respectively. Finally, we define $\V u_{n,k} \in \R^d$ as the unit vector that is normal to $L_{n,k}$. We follow these notation throughout the paper. 
 
Alternatively, any CPWL function can be defined by a triangulation and the values of the function at its vertices, which we refer to as the simplicial parameterization of CPWL functions. The Authors in \cite{liu2020delaunay} use a similar parameterization. This parameterization yields a Riesz basis, which guarantees a unique and stable link between the model parameters and the CPWL function \cite{goujon2022stable}. In Figure 1, we show an example of the domain of an arbitrary CPWL function and a possible triangulation of it.

\begin{figure}[h]
\begin{center}
  \includegraphics[width=8cm,height=5cm,keepaspectratio]{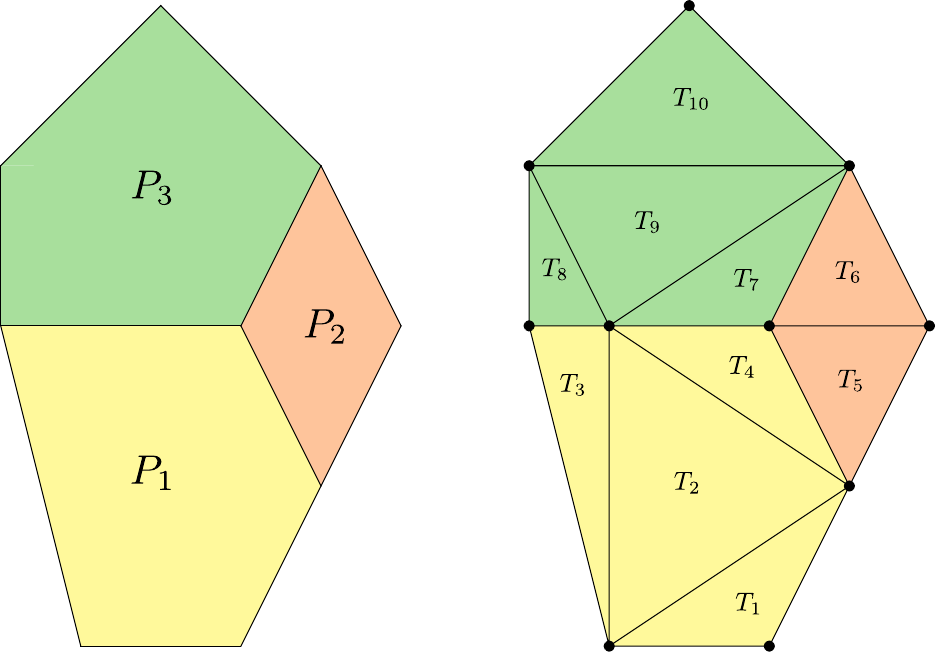}
  \caption{Domain of a two-dimensional CPWL function. Each patch corresponds to one linear region (left) and its partition into simplices through a Delaunay triangulation (right).}
  \label{fig:delcpwl}
\end{center}
\end{figure}
\subsection{Hessian Total Variation}
\subsubsection{Generalized Hessian Matrix}
The Hessian matrix of a twice-differentiable function $f: \R^d \rightarrow \R$ is defined as
\begin{equation}
    \Op H\{f\} = \begin{bmatrix}
            \frac{\partial^2 f}{\partial x_1^2} & \cdots & \frac{\partial^2 f}{\partial x_1 \partial x_d}  \\
            \vdots & \ddots & \vdots \\
            \frac{\partial^2 f}{\partial x_d \partial x_1} & \cdots &  \frac{\partial^2 f}{\partial x_d^2}
    \end{bmatrix}.
\end{equation}
One can extend the Hessian definition to generalized functions (distributions) by using the notion of weak partial derivative. This enables us to define the Hessian of CPWL functions even though they are not twice-differentiable everywhere. The directional derivative at a point $\V x$ and along the unit vector $\V u$ is defined as 
\begin{equation}
    \Dop_{\V u} f(\V x) = \lim_{h\to0} \frac{f(\V x + h \V u) - f(\V x)}{h}.
\end{equation}
Likewise, the second-order directional derivative along the direction $\V u$ is $\Dop_{\V u}^2 f = \Dop_{\V u} \circ \Dop_{\V u} (f)$. Using the Hessian matrix, we write that $\Dop_{\V u}^2 f(\V x) = \V u^T \Op{H}\{f\}(\V x) \V u$. A symmetric Hessian matrix has a complete set of eigenvalues and eigenvectors which form an orthogonal basis for $\R^d$. Consequently, the second-order directional derivative along the eigenvector $\V v_q$ of the Hessian is its associated eigenvalue $\lambda_q$, with $\Dop_{\V v_q}^2 f(\V x) = \lambda_q$ for $q \in {1, \ldots, d}$. If we use the eigenvectors of the Hessian to represent the direction $\V u = \sum_{q=1}^{d} t_q \V v_q$, then we have that $\Dop_{\V u}^2 f(\V x) = \sum_{q=1}^{d} t_q^2 \lambda_q$. This means that, at each point $\V x \in \R^d$ of the domain, the second-order directional derivatives along the eigenvectors of the Hessian fully characterize the second-order derivatives of $f$ along any direction.

\subsubsection{Schatten p-Norm}
The Schatten p-norm $\norm{\cdot}_{{\mathcal{S}}_p}$ of a matrix $\M A \in \R^{d \times d}$ for $p \in [1, +\infty]$ is defined as the $\ell_p$-norm of its singular values given by
\begin{equation}
    \norm{\M A}_{{\mathcal{S}}_p} \coloneqq \begin{cases}
    {\left(\sum\limits_{k=1}^{d} {\sigma_k}^p\right)}^{\frac{1}{p}}, & 1 \leq p < +\infty\\
    \max\limits_{k} {\sigma_k}, & p=+\infty,
    \end{cases}
\end{equation}
where $(\sigma_1, \ldots, \sigma_d)$ are the singular values of $\M A$. In this paper, we focus on the $\mathcal{S}_1$-norm and its dual $\mathcal{S}_{\infty}$.  
\subsubsection{Hessian Total Variation}
If $f$ is twice differentiable with a symmetric Hessian at $\V x \in \R^d$, then $\norm{\Op{H}\{f\}(\V x)}_{\mathcal{S}_1}$ is given by the $\ell_1$-norm of the second-order directional derivatives along the eigenvectors of the Hessian matrix. This measure provides a local characterization of the second-order variation of $f$. Hence, the total variation of the mapping $\cdot \mapsto \norm{\Op{H}\{f\}(\cdot)}_{\mathcal{S}_1}$ is a reasonable generalization of the second-order total variation for multidimensional functionals. This is referred to as the Hessian total variation ($\HTV$). If $\Op{H}\{f\} \in L_1(\R^d; \R^{d\times d})$, then $\HTV$ is simply defined as

\begin{align}
    {\HTV}(f) = \norm{\norm{\Op{H}\{f\}}_{\mathcal{S}_1}}_{L_1} =
    \notag \\ \int_{\R^d}
   \sum\limits_{k=1}^{d} \sigma_k(\Op{H}\{f\}(\V x)) {\mathrm d} \V x.
\end{align}

In particular, by using the $\HTV$ as a regularizer, we promote configurations where most of the singular values of the Hessian are zero throughout the function domain due to the sparsity effects of the $\ell_1$- and $L_1$-norms. In that manner, CPWL functions are of special interest as their second-order directional derivatives vanish almost everywhere. Specifically, the gradient of a CPWL function $f$ admits the explicit form
\[
  {\bm \nabla} f(\V x) = \sum_{n} \V a_n \Indic_{P_n}(\V x),
\]{gradient_CPWL}
where the indicator function $\Indic_{P_n}$ is equal to one inside the simplex $P_n$ and zero elsewhere. Hence, the second-order directional derivatives vanish everywhere except on the boundaries of the domain polytopes. There, the generalized Hessian matrix of the CPWL functions exhibits delta-like distributions. This means that (8) does not hold for CPWL functions as $\Op{H}\{f\} \notin L_1(\R^d; \R^{d\times d})$. Yet, we can extend the definition of the Hessian total variation to accommodate CPWL functions by invoking duality. Specifically, we define
\begin{equation}
    {\rm HTV}(f) = \norm{\Op H \{f\}}_{\mathcal{S}_1, \mathcal{M}},
\end{equation}
where the mixed norm $ \norm{\cdot}_{\mathcal{S}_1, \mathcal{M}}$ for any matrix-valued distribution $W \in \mathcal{S}^{\prime}(\R^d; \R^{d\times d})$ is
\begin{align}
     \norm{\Op W}_{\mathcal{S}_1, \mathcal{M}} \coloneqq \sup\{\langle W, F \rangle : F \in \mathcal{S}(\R^d; \R^{d\times d}), \notag \\ 
     \norm{F}_{\mathcal{S}_\infty} \leq 1\}.
\end{align}
In (11), $\mathcal{S}(\R^d; \R^{d\times d})$ is the matrix-valued Schwartz space and the duality product is defined as
\begin{equation}
    \langle W, F \rangle \coloneqq \sum_{p=1}^{d} \sum_{q=1}^{d} \langle w_{p, q}, f_{p, q} \rangle.
\end{equation}
This formulation enables us to derive a closed form for the $\HTV$ of a CPWL function $f$ as
\begin{align}
    {\rm HTV}(f) &=
     \sum_{(n, k) \in \mathcal{N}} \abs{\V u_{n,k}^T (\V a_k - \V a_n)} {\rm Vol}_{d-1}(L_{n,k})  \\
    & =
     \sum_{(n, k) \in \mathcal{N}}  \norm{\V a_k - \V a_n}_2  {\rm Vol}_{d-1}(L_{n,k}),
  \end{align}
where the set $\mathcal{N}$ contains all unique pairs of neighboring polytope indices (see \cite{aziznejad2021measuring} for more details). Equation (14) tells us that, for CPWL functions, $\HTV$ penalizes the change of slope at neighboring domain polytopes by the volume of their intersection. For a purely affine mapping, the $\HTV$ is zero. Otherwise, it increases with the number of affine pieces.

$\HTV$ admits the following properties:
\begin{align}
  & \text{1)} \quad \HTV\left(f(\M U \, \cdot))\right) 
  = 
  \HTV(f),  \quad \M U \in \R^{d\times d} : \M U^T \M U = \M I; \notag\\
  & \text{2)} \quad \HTV\left(f(\alpha \, \cdot)\right) 
  = 
   \,  \abs{\alpha} ^{2-d}\HTV(f), \quad \quad \forall \alpha \in \R; \notag\\
  & \text{3)} \quad \HTV\left(f(\cdot - \V x_0)\right) 
  = 
  \HTV(f), \quad \quad \forall \V x_0 \in \R^d.
\label{eq:rotation_scale_translation_invariance}
\end{align}
In other words, $\HTV$ is invariant to translation and rotation while it is covariant with scaling. 
\section{Methods}
Let us perform a Delaunay triangulation on the set $\M T_g = \{\V \tau_k\}_{k=1}^{N_g} \subset \R^d$ of grid points. We denote the resulting simplices by $\mathcal{T}=\{T_n\}_{n=1}^{N_s}$. We define the CPWL functions $f_{\mathrm{Hull}}$ inside the convex hull $\mathrm{Hull}(\M T_g)$ by designating the members of $\mathcal{T}$ as their linear regions. Then, the functions are parameterized as 

\begin{equation}
    f_{\mathrm{Hull}}(\V x) = \sum_{k = 1}^{N_g} c_k s_k(\V x)
\end{equation}
with expansion coefficients $c_k \in \R$ and basis functions $s_k$. The basis $s_k$ is the hat function attached to the $k$th vertex. It is is given by
\begin{equation}
    s_k(\V x) = \begin{cases}
        \beta_{k, l}(\V x), & \V x \in T_{k, l} \\     
        0, & \text{otherwise}.
    \end{cases}
\end{equation}
$T_{k,l}$ is the $l$-th simplex that contains the vertex $k$ and $\beta_{k, l}(\V x)$ is the barycentric coordinate of the point $\V x$ inside the simplex $T_{k,l}$ with respect to vertex $k$. The basis functions $s_k$ are continuous and form a Riesz basis for the space. From the definition of barycentric coordinates, we know that $s_k(\V \tau_k) = 1$ and that $s_k(\V \tau_n) = 0, n \neq k$. We illustrate an example of this basis function in Figure 2 for $d=2$. Consequently, the $c_k$ in (16) are given by $c_k = f_{\mathrm{Hull}}(\V \tau_k)$. The item of relevance is that the function $f_{\mathrm{Hull}}$ is uniquely identified by $\{(\V \tau_k, c_k)\}_{k=1}^{N_g}$. An example of such a function is illustrated in Figure 3 (a). In our work, the grid points are immutable and $\V c = (c_k)$ is our vector of learnable parameters. These parameters are the sampled values of the function at the vertices of the triangulation. Hence, their effect on the generated mapping is known, which makes our model directly interpretable. We can also write the function $f_{\mathrm{Hull}}$ as
\begin{equation}
    f_{\mathrm{Hull}}(\V x) = \V \gamma_{\V x}^T \V c,
\end{equation}
where $\V \gamma_{\V x} = (\V s_k(\V x))_{k=1}^{N_g}$. The importance of this representation is that it expresses the function value as a linear combination of the grid-point values $\V c$. In each simplex, the only nonzero basis functions are the ones defined over the vertices of that simplex. This implies that there are at most ($d+1$) nonzero values in the vector $\V \gamma_{\V x}$. 
\begin{figure}[h]
\begin{center}
  \includegraphics[width=8cm,height=5cm,keepaspectratio]{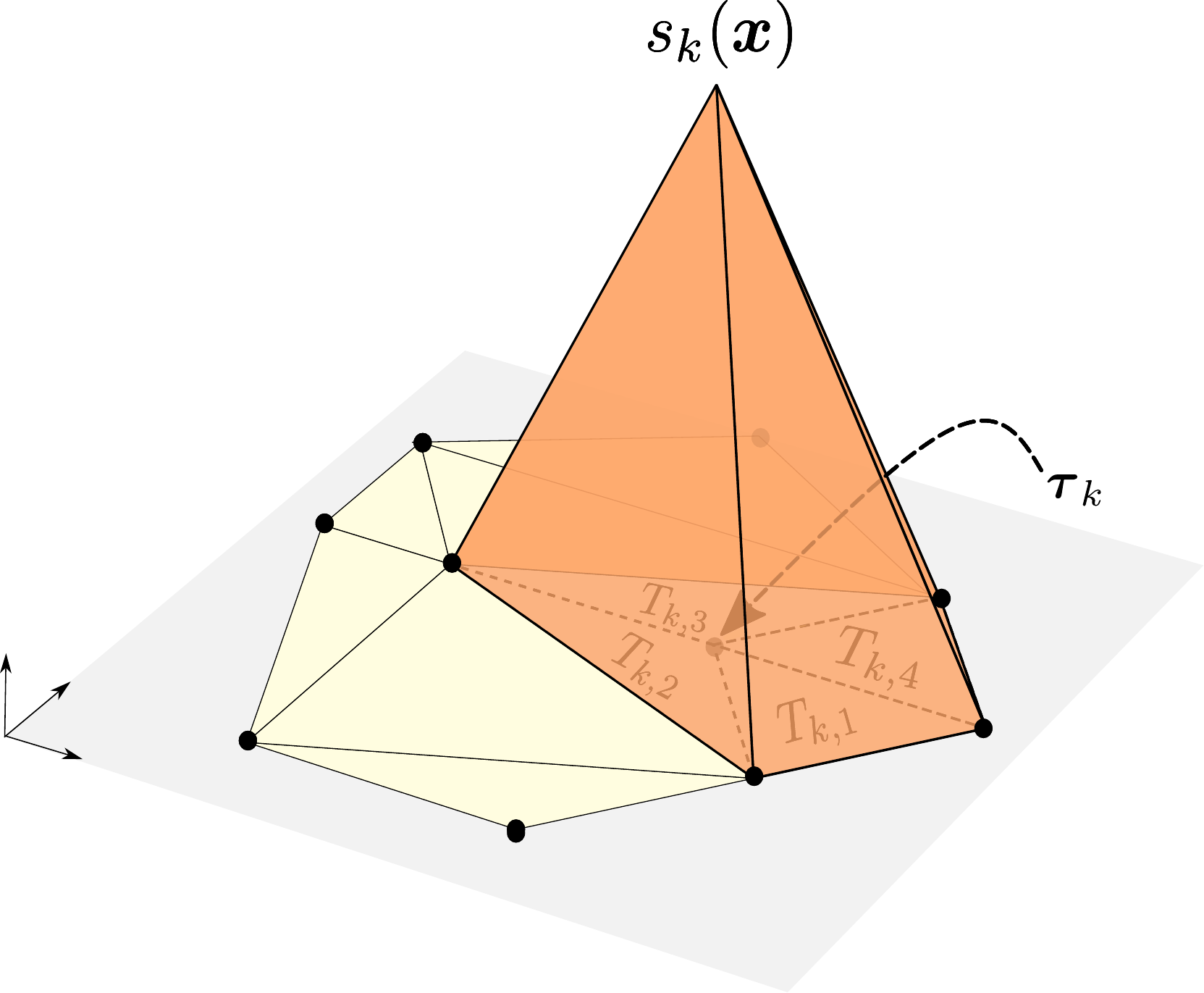}
  \caption{The $k$th vertex of the triangulation with coordinate $\V \tau_k$ is included in 4 domain simplices $\{T_{k,l}\}_{l=1}^4$ over which the associated simplicial basis (hat) function $s_k(\V x)$ is illustrated.}
  \label{fig:delcpwl}
\end{center}
\end{figure}

\begin{figure}[h]
\begin{center}
  \subfloat[]{\includegraphics[width=8cm,height=5cm,keepaspectratio]{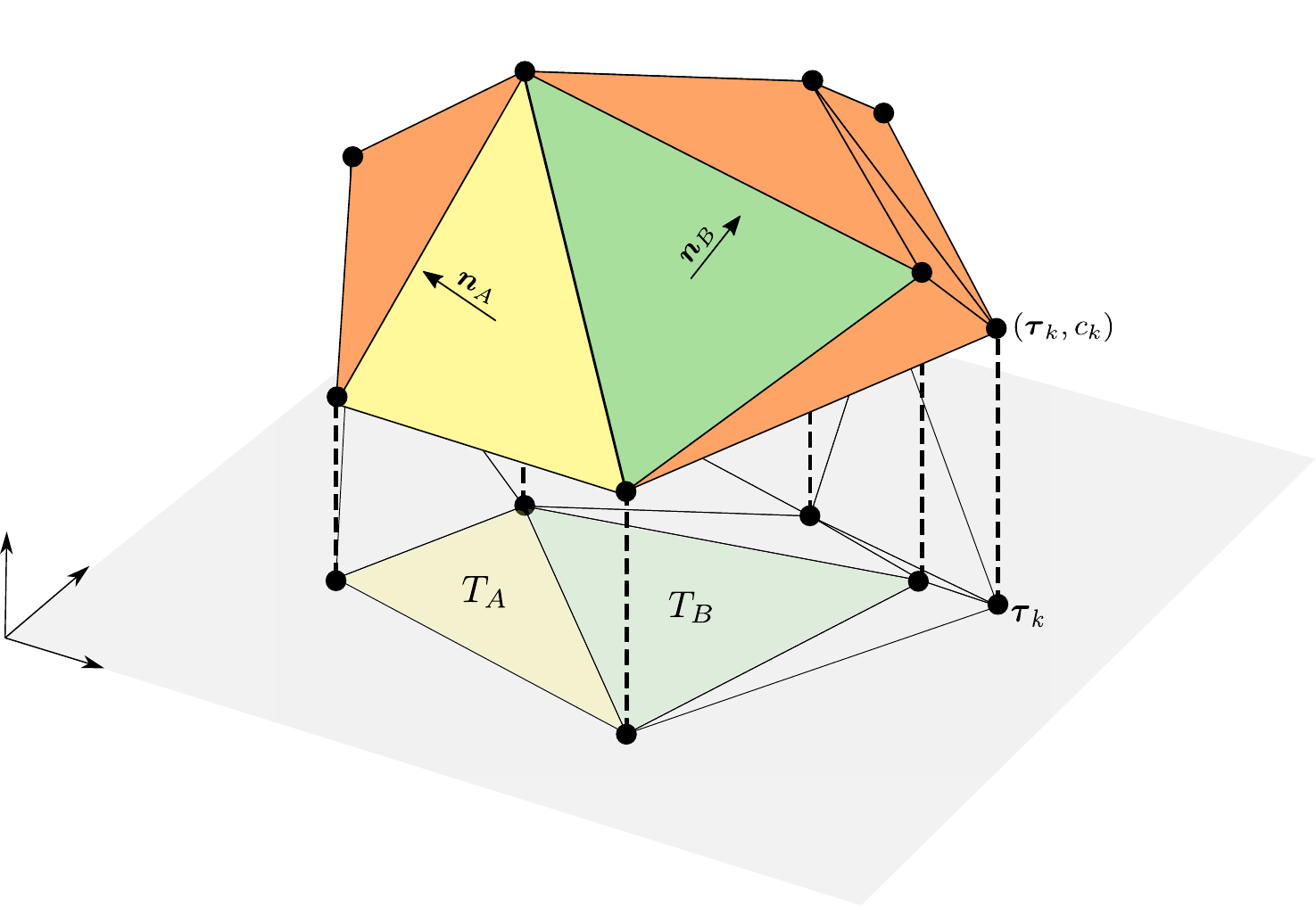}}

  \subfloat[]{\includegraphics[width=8cm,height=5cm,keepaspectratio]{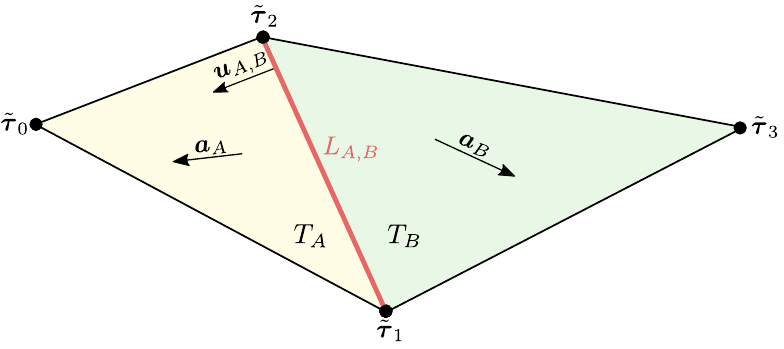}}
  \caption{(a) Example of a CPWL function $f_{\mathrm{Hull}}$ defined through a triangulation. Over each simplex $T_A$, one can identify the function $f_{\mathrm{Hull}}$ by the plane that passes through the points $\{(\V \tau_k, f_{\mathrm{Hull}}(\V \tau_k) = c_k)\}_{\V \tau_k \in T_A}$. The vector $\V n_A$ corresponds to the normal vector of that plane.
  (b) Top view of the domain simplices $T_A$ and $T_B$.  
  }
  \label{fig:delcpwl}
\end{center}
\end{figure}
\subsection{Forward Operator}

Given the set of training $\{(\V x_m, y_m)\}_{m=1}^{M}$ data, we choose grid points $\M T_g$ such that all data points are inside $\mathrm{Hull}(\M T_g)$. Then, the function $f_{\mathrm{Hull}}$ is evaluated at data points as
\begin{equation}
    {\V f} = \M H \V c,
\end{equation}
where ${\V f} = (f_m)_{m=1}^{M}$ and $f_m= f_{\mathrm{Hull}}(\V x_m)$. The matrix $\M H \in \R^{M \times N_g}$ is referred to as  the forward operator and is a mapping between the grid-point values $\V c$ and the values of the function $\V f$ at the data points. If we represent the forward matrix as 
\begin{equation}
    \M H = \begin{bmatrix}
            \V h_1^T \\
            \vdots \\
            \V h_{N_g}^T
    \end{bmatrix},
\end{equation}
then, by (18), each row of $\M H$ can be written as $\V h_m =  \V \gamma_{\V x_m} $. This implies that $\M H$ is a sparse matrix with at most ($d+1$) nonzero entries per row. 

\subsection{Regularization Operator}
To determine the Hessian total variation of the function $f_{\mathrm{Hull}}$ through (13), we need to calculate three quantities for each pair of neighbor simplices in $\mathcal{T}$; the gradient difference of the affine pieces over those pairs; the unit normal of their intersection; and, finally, their intersection volume.

\noindent \textbf{Setup:}  Assume that $T_A, T_B \in \mathcal{T}$ are two neighboring simplices and that the sets of indices of their vertices are $\mathcal{V}_A$ and $\mathcal{V}_B$, with $\abs{\mathcal{V}_A} = \abs{\mathcal{V}_B} = d+1$. We denote the vertices of their intersection by $\mathcal{V}_{A,B} = \mathcal{V}_A \cap \mathcal{V}_B$. There are exactly $d$ common vertices between two neighboring simplices, so that $\abs{\mathcal{V}_{A,B}} = d$. We assume that $\tilde{\V \tau}_0,  \tilde{\V \tau}_{d+1}$ and $\tilde{c}_0, \tilde{c}_{d+1}$ are the coordinates and values at the vertices indexed by members of $\mathcal{V}_A \setminus \mathcal{V}_{A,B}$ and $\mathcal{V}_B \setminus \mathcal{V}_{A,B}$. Also, we denote $\tilde{\V \tau}_l$ and $\tilde{c}_l$ the coordinate and value of the vertex indexed by the $l$-th smallest member of $\mathcal{V}_{A,B}$ for $1 \leq l \leq d$. We show an example of this setup in Figure 3 (b).

\begin{theorem}[Gradient difference]
Let $\V \nabla f_{\mathrm{Hull}}|_{T_A}(\V x) = {\V a}_A$ and $\V \nabla  f_{\mathrm{Hull}}|_{T_B}(\V x) = {\V a}_B$. Their difference can be expressed as the linear combination of grid points values $\V c$ given by

\begin{equation}
    {\V a}_A - {\V a}_B = {\M G}_{A,B} {\V c},
\end{equation}
where
\begin{equation*}
    {\M G}_{A,B} = \begin{bmatrix} \M G_A \V 1 & \M G_B - \M G_A & -\M G_B\V 1 \end{bmatrix} \M W_{A, B}, \\[11pt]
\end{equation*}
\begin{equation}
    \M G_A = \begin{bmatrix}
                (\tilde{\V \tau}_0 - \tilde{\V \tau}_1)^T \\
                \vdots \\
                (\tilde{\V \tau}_0 - \tilde{\V \tau}_{d})^T
             \end{bmatrix}^{-1}, \M G_B= \begin{bmatrix}
                (\tilde{\V \tau}_{d+1} - \tilde{\V \tau}_1)^T \\
                \vdots \\
                (\tilde{\V \tau}_{d+1} - \tilde{\V \tau}_d)^T
             \end{bmatrix}^{-1}.
\end{equation}
There, the symbol $\V 1$ represents the vector $(1)_{k=1}^{d}$ and $\M W_{A, B} = [w_{p, k}]_{(d+2)\times N_g}$ is a sparse binary matrix such that $w_{p,k}= 1$ if and only if $\tilde{\V \tau}_p = {\V \tau}_k$.
\end{theorem}
\begin{proof}
Since $f_{\mathrm{Hull}}$ is affine over the simplices $T_A$ and $T_B$, we have that
\begin{equation}
    \begin{cases}
        (\tilde{\V \tau}_0 - \tilde{\V \tau}_p)^T {\V a}_A  = \tilde{c}_0 - \tilde{c}_p, & p= 1,...,d \\
        (\tilde{\V \tau}_{d+1} - \tilde{\V \tau}_p)^T {\V a}_B = \tilde{c}_{d+1} - \tilde{c}_p, & p = 1,...,d.
    \end{cases}
\end{equation}
Putting all equations together, we obtain that
\begin{align}
     {\V a}_A = \begin{bmatrix}
                (\tilde{\V \tau}_0 - \tilde{\V \tau}_1)^T \\
                \vdots \\
                (\tilde{\V \tau}_0 - \tilde{\V \tau}_{d})^T
             \end{bmatrix}^{-1} \begin{bmatrix}
            \tilde{c}_0 - \tilde{c}_1 \\
            \vdots \\
            \tilde{c}_0 - \tilde{c}_d 
         \end{bmatrix} = \M G_A \begin{bmatrix}
            \tilde{c}_0 - \tilde{c}_1 \\
            \vdots \\
            \tilde{c}_0 - \tilde{c}_d 
         \end{bmatrix} \notag \\[10pt]
        = \tilde{c}_0 \M G_A \V 1 - \M G_A \begin{bmatrix}
            \tilde{c}_1 \\
            \vdots \\
            \tilde{c}_d 
         \end{bmatrix}.
\end{align}
By analogy, we have that
\begin{equation}
         {\V a}_B = \tilde{c}_{d+1} \M G_B \V 1 - \M G_B \begin{bmatrix}
            \tilde{c}_1 \\
            \vdots \\
            \tilde{c}_d 
         \end{bmatrix}.
\end{equation}
Next, we write the difference of $\V a_A$ and $\V a_B$ as
\begin{equation}
     \V a_A = \begin{bmatrix} \M G_A \V 1 & -\M G_A & \V 0 \end{bmatrix} \begin{bmatrix}
            \tilde{c}_0 \\
            \vdots \\
            \tilde{c}_{d+1}
         \end{bmatrix} ,
\end{equation}

\begin{equation}
     \V a_B = \begin{bmatrix} \V 0  & -\M G_B & \M G_B \V 1 \end{bmatrix} \begin{bmatrix}
            \tilde{c}_0 \\
            \vdots \\
            \tilde{c}_{d+1}
         \end{bmatrix} ,
\end{equation}
where $\V 0 = (0)_{k=1}^{d}$.
 \begin{figure*}[h]
  \begin{center}
  \includegraphics[keepaspectratio, width=\textwidth]{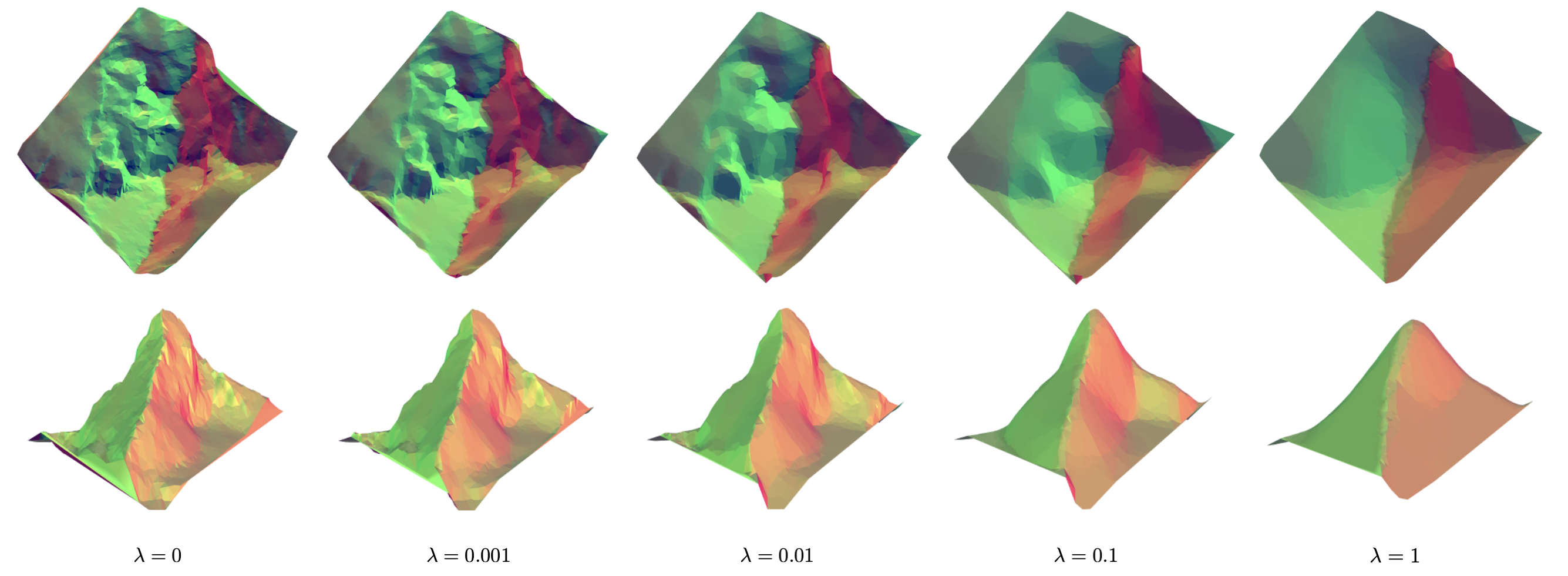}
  \caption{Effect of the regularization hyperparameter $\lambda$ on the learned function for the evaluation map of the Matterhorn with our framework. Each column includes the results for the specified $\lambda$. The images in the first and second row give the top and side views of the learned mappings.}
  \label{fig:Cervin}
  \end{center}
\end{figure*}
Hence, 

\begin{equation}
    {\V a}_A - {\V a}_B = \begin{bmatrix} \M G_A \V 1 & \M G_B - \M G_A & -\M G_B\V 1 \end{bmatrix} \begin{bmatrix}
            \tilde{c}_0 \\
            \vdots \\
            \tilde{c}_{d+1}
         \end{bmatrix}
\end{equation}
and, from the definition of $\M W_{A, B}$, we know that
\begin{equation}
    \begin{bmatrix}
            \tilde{c}_0 \\
            \vdots \\
            \tilde{c}_{d+1}
         \end{bmatrix} = \M W_{A, B} \V c.
\end{equation}
From (28) and (29), we finally obtain (21).
 \end{proof}
\begin{theorem}[Determination of a normal vector]
The unit normal ${\V u}_{A,B}$ of the intersection of $T_A$ and $T_B$ is given by
  \begin{equation}
    {\V u}_{A,B} = \frac{\M N_{[1:d,1]}}{\norm{\M N_{[1:d,1]}}},
 \end{equation}
 where
\begin{equation}
\M N = 
        \begin{bmatrix}
            \tilde{\V \tau}_0^T & 1 \\
            \tilde{\V \tau}_1^T & 1 \\
            \vdots & \vdots \\
            \tilde{\V \tau}_{d}^T & 1
        \end{bmatrix}^{-1}
    \end{equation}
and where the slicing operator $[1:d, 1]$ returns the first $d$ elements of the first column of $\M N$.
\end{theorem}
\begin{proof}
The intersection of $T_A$ and $T_B$ is the facet opposite to the vertex with coordinate $\tilde{\V \tau}_0$ in the simplex $T_A$. We show that $\V z_0 = \M N[1:d, 1]$  is perpendicular to that facet and, hence, to the intersection. From the definition of the inverse, we know that

 \begin{equation}
    \begin{bmatrix}
            \tilde{\V \tau}_0^T & 1 \\
           \tilde{\V \tau}_1^T & 1 \\
            \vdots & \vdots \\
            \tilde{\V \tau}_d^T & 1
        \end{bmatrix} \M N = \M I. 
\end{equation}
Hence, if we write the matrix $\M N$ as 

\begin{equation}
    \M N = \begin{bmatrix}
                \V z_0 & \cdots & \V z_d \\
                -b_0 &  \cdots & -b_d 
        \end{bmatrix} 
\end{equation}
for some $\{\V z_k\}_{k=0}^{d}  \subset \R^d$ and $\{b_k\}_{k=0}^d \subset \R$, then we have that
 \begin{equation}
    \begin{bmatrix}
            \tilde{\V \tau}_0^T & 1 \\
            \vdots & \vdots \\
            \tilde{\V \tau}_d^T & 1
        \end{bmatrix} 
        \begin{bmatrix}
                \V z_0 & \cdots & \V z_d \\
                -b_0 &  \cdots & -b_d 
        \end{bmatrix} = \begin{bmatrix}
            1 &  & 0 \\
 
            &  \ddots &  \\
            0  & & 1        \end{bmatrix}.
\end{equation}
In particular, this implies that $\tilde{\V \tau}_0^T \V z_0 = 1 + b_0$ and $\tilde{\V \tau}_m^T \V z_0 = b_0$ for $m > 0$, so that
\begin{equation}
     (\tilde{\V \tau}_m - \tilde{\V \tau}_l)^T \V z_0 = b_0 - b_0 = 0,
\end{equation}
for $m > 0, l \neq m$. Equation (35) implies that $\V z_0$ is perpendicular to the simplex formed by $\tilde{\V \tau}_1, \ldots, \tilde{\V \tau}_d$, which is exactly the facet opposite to $\tilde{\V \tau}_0$ and is the intersection of $T_A$ and $T_B$.
\end{proof}
\begin{theorem}[Cayley–Menger determinant \cite{blumenthal1943distribution}]
 The ($d$-$1$) dimensional volume $\mathrm{Vol}$ of the simplex formed by $\{\tilde{\V \tau}_1, \ldots,  \tilde{\V \tau}_d\}$ is given by  
 
\begin{equation}
   \mathrm{{Vol}}^2  = \gamma \begin{vmatrix}
    0 & \tilde{d}_{1,2} &   \cdots &  \tilde{d}_{1,d} & 1 \\
     \tilde{d}_{2,1} & 0 & \cdots & \tilde{d}_{2,d} & 1 \\
    \vdots & \vdots &  \ddots & \vdots & \vdots \\
     \tilde{d}_{d,1} &  \tilde{d}_{d,2} &   \cdots & 0 & 1 \\
    1 & 1 & \cdots & 1 & 0
    \end{vmatrix},
\end{equation}
where $\tilde{d}_{k,l} = \norm{\tilde{\V \tau}_k - \tilde{\V \tau}_l}^2$ and
\begin{equation}
   \gamma =  \frac{(-1)^{d}}{((d-1)!)^2 2^{d-1}}.
\end{equation}
\end{theorem}

The intersection of $T_A$ and $T_B$ is the simplex formed by $\{\tilde{\V \tau}_1, \ldots, \tilde{\V \tau}_d\}$. Hence, by using Theorem 3, we can obtain the intersection volume $\mathrm{{Vol}}_{A,B}$. \\

We define the matrix $\M R_{A,B}$ as
\begin{equation}
    \M R_{A,B} = \mathrm{{Vol}}_{A,B} {\V u}_{A,B}^T {\M G}_{A,B}.
\end{equation}
From Theorems 1-3, we have that
\begin{equation}
    \abs{\V u_{A,B}^T (\V a_A - \V a_B)} {\rm Vol}_{d-1}(L_{A,B}) = \abs{\M R_{A, B} \V c}.
\end{equation}
From (13), we then calculate the $\HTV$ of $f_{\mathrm{Hull}}$ as
\begin{equation}
    \HTV(f_{\mathrm{Hull}}) = \sum_{(A, B) \in \mathcal{N_\mathcal{T}}} \abs{\M R_{A,B} \V c} = \norm{\M L \V c}_1, 
\end{equation}
where $\mathcal{N_\mathcal{T}}$ is the set of all unique neighbor simplices in $\mathcal{T}$. In (40), $\M L$ is a sparse matrix that is referred to as the regularization operator. It is of size $(\abs{\mathcal{N}_\mathcal{T}}, N_g)$ and each of its rows corresponds to the $\HTV$ term associated with two neighbor simplices. Hence, there are at most $d+2$ nonzero elements at each row of $\M{L}$. 
\subsection{Learning Problem}
To learn the regressor $\hat{f}$ for some given data $\{\V x_m, y_m\}_{m=1}^{M}$, we propose to solve the optimization problem
\begin{equation}
  \hat{f} \in \argmin\limits_{f \in \mathcal{X_{\mathrm{CPWL}}}} \sum_{m=1}^M \Op{E}(f(\V x_m), y_m) + \lambda \HTV(f),
\label{htv_learning}
\end{equation}
where $\mathcal{X_{\mathrm{CPWL}}}$ represents the CPWL function search space. If we restrict the search space to CPWL functions $f_{\mathrm{Hull}}$ parameterized by a triangulation over a set of grid points $\M T_g$, then we have that
\begin{equation}
   \hat{f}_{\mathrm{Hull}} \in \argmin\limits_{f_{\mathrm{Hull}} \in \mathcal{X}_{\M T_g}}  \sum_{m=1}^M \Op{E}(f_{\mathrm{Hull}}(\V x_m), y_m) + \lambda \HTV(f_{\mathrm{Hull}}).
\label{htv_learmning_2}
\end{equation}
Using the forward and regularization operators introduced in Section III.A and III.B and with $\V y = (y_m)_{m=1}^{M}$, we rewrite (42) as
\begin{equation}
    \hat{\V c} \in \argmin\limits_{\V c \in \R^{N_g}} \frac{1}{2} \norm{\V y - \M H \V c}_2^2 + \lambda \norm{\M L \V c}_1.
\end{equation}
Formulation (43) recasts the problem of finding $\hat{f}$ into a discrete problem of finding grid values $\hat{\V c}$. It is generically referred to as the generalized LASSO in the literature. The problem is convex and has solutions that can be found using methods such as the alternating-direction method of multipliers (ADMM)  \cite{tibshirani2011solution}, \cite{boyd2011distributed}. In the special case when the forward operator is the identity $\M H = \M I$, the dual of the problem is
\begin{equation}
    \hat{\V u} \in \argmin\limits_{\V u \in \R^{\abs{\mathcal{N}_\mathcal{T}}}} \frac{1}{2} \norm{\V y - \M L^T \V u}_2^2 \text{ subject to } \norm{\V u}_{\infty} \leq \lambda,
\end{equation}
where the relation $\hat{\V c} = (\V y - \M L^T \hat{\V u})$ holds. Although the dual problem (44) is high-dimensional, it is proximable. This means that we can use optimization algorithms such as fast iterative shrinkage-thresholding algorithm (FISTA) to solve it whereas they are not applicable for the primal problem \cite{beck2009fast}. FISTA has better convergence rate than ADMM and helps us to accelerate our computations. Algorithm 1 describes the iterations of FISTA when solving (44). The initialization $\V u_0 = \V 0$ is equivalent to the interpolation state ($\lambda = 0$) as it corresponds to $\hat{\V c} = \V y$. The value of $\alpha$ is less than the inverse of the largest eigenvalue of $2 \M{L}^T \M{L}$ and the function $\textbf{Clip}(\V a, \lambda)$ is defined as
\begin{equation}
    (\textbf{Clip}(\V a, \lambda))_k = \begin{cases}
    -\lambda, & a_k < -\lambda \\
    a_k, &, \abs{a_k} \leq \lambda \\
    \lambda, & {a_k} > \lambda.
\end{cases}
\end{equation}

\begin{table*}[h]
\begin{center}

    \renewcommand{\arraystretch}{1.3}
 \caption{Average training and testing mean-square error, normalized $\HTV$, number of parameters, and sparsity metric for several learning approaches used applied to the power-plant data set.}
\label{table:4dresults}
\begin{tabular}{lrrrrrrl} 

\toprule
 & Train MSE & Test MSE & Number of parameters & HTV & Sparsity  \\ 
\midrule
LR                                                                                                 &    $20.85 $        & $ 20.71 $   & $5 $ &   $0.0$   &  $100 \% $       \\ 
 DHTV                                                                            &  $1.04 $              & $\mathbf{13.14} $        & $5725$   &                 $0.75$   &     $65 \%$     \\ 
NN2                                                                                                               &  $ 11.17$            &      $14.25 $          &      $253501$      &            $1.92 $ &    $58 \%$       \\ 
 NN6                                                                             &         $10.98$         &       $14.24
$           &     $ 202201 $        &        $2.29$  &     $51 \%$      \\ 
NN36                                                                                                              &        $3.06$         &     $15.50 $            &     $ 1408201$  & $2.31 $  &     $37 \%$     \\ 
 RBF                                                                             &  $12.16$          &       $15.85$          &     6697                         &   $21.47$  &      $ 42 \%$\\
\bottomrule
\end{tabular}
 \renewcommand{\arraystretch}{1}

\end{center}

\end{table*}
\begin{algorithm}
\caption{FISTA iterations to solve (43)}\label{alg:cap}
\begin{algorithmic}
\State $\textbf{Initialize}: {\V u_0 = \V 0, \V v_0 = \V 0, t_k = 1} $\;
\For{$k = 0 \rightarrow N_{\mathrm{iter}}$}
\State $\V u_{k+1} \leftarrow \textbf{Clip}(\V v_k - \alpha (-2 \M L \V y + 2 \M L \M L^T \V v_k), \lambda)$\;
\State $t_{k+1} = \frac{1 + \sqrt{4 t_k^2 + 1}}{2}$\;
\State $\V v_{k+1} = \V u_{k+1} + \frac{(t_k - t_{k+1})}{t_{k+1}}  (\V u_{k+1} - \V u_k)$
\EndFor
\end{algorithmic}
\end{algorithm}

\subsection{Final Regressor}
The solution of (43) gives us the grid values $\hat{\V c}$ that define uniquely the CPWL function $\hat{f}_{\mathrm{Hull}}$. However, the domain of the definition of this function is restricted to the convex hull of the grid points. The most intuitive way to extend $\hat{f}_{\mathrm{Hull}}$ is to use the notion of nearest neighbors and assign the value of the closest grid points to the points outside the convex hull. However, the mapping generated by using nearest neighbors is piecewise-constant outside the convex hull, which does not generate a global CPWL relation. To overcome this issue, we define our final CPWL function $\hat{f}_{\mathrm{CPWL}}$ over $\R^d$ as
\begin{equation}
    \hat{f}_{\mathrm{CPWL}} (\V x) = \begin{cases}
        \hat{f}_{\mathrm{Hull}}(\V x), & \V x \in \mathrm{Hull}(\M X_g) \\
       \hat{f}_{\mathrm{Hull}} (\Phi (\V x)), & \text{otherwise}, 
    \end{cases}
\end{equation}
where $\Phi$ is the orthogonal projection of the point $\V x$ onto the convex hull of the grid points. The projection on the convex hull can be formulated as a quadratic optimization problem \cite{gabidullina2018problem}. It can be shown that the solution of this problem corresponds to CPWL functions, so that our final regressor is guaranteed to be globally CPWL \cite{spjotvold2007continuous}.

\section{Experiments}
In this section, we choose the grid points as the data points. In case of duplicate data points, we keep only one of them and let its target value to be the average of the target values of the duplicates. This results in the forward operator $\M H$ being identity and enables us to solve the learning problem with \textbf{Algorithm 1}. To have comparable ranges in each dimension of the input space, we standardize each feature. For the Delaunay triangulation, we use Scipy, which can safely handle data up to dimension $d=9$ \cite{2020SciPy-NMeth}. Our codes are available on the GitHub repository\footnote{https://github.com/mehrsapo/DHTV}. 

\subsection{Matterhorn}
The purpose of our first experiment is to illustrate the effect of regularization. To that end, we consider the elevation map $f: \R^2 \to \R$ of the iconic Swiss mountain Matterhorn. We sample randomly 4800 points from the domain of the function\footnote{ASTER Global Digital Elevation Model V003: https://lpdaac.usgs.gov/products/astgtmv003/}. Then, the mountain is reconstructed from the sampled data using our framework, which we refer to as Delaunay Hessian total variation (DHTV). The result of the DHTV learning is reported in Figure 4 for several values of the regularization hyperparameter $\lambda$. When $\lambda$ increases, we see that the number of affine pieces of the final mapping decreases. Interestingly, the regularizer tends to omit local fluctuations while preserving the main ridges. 
\begin{figure}[h]
\begin{center}
  \includegraphics[width=8cm,height=5cm,keepaspectratio]{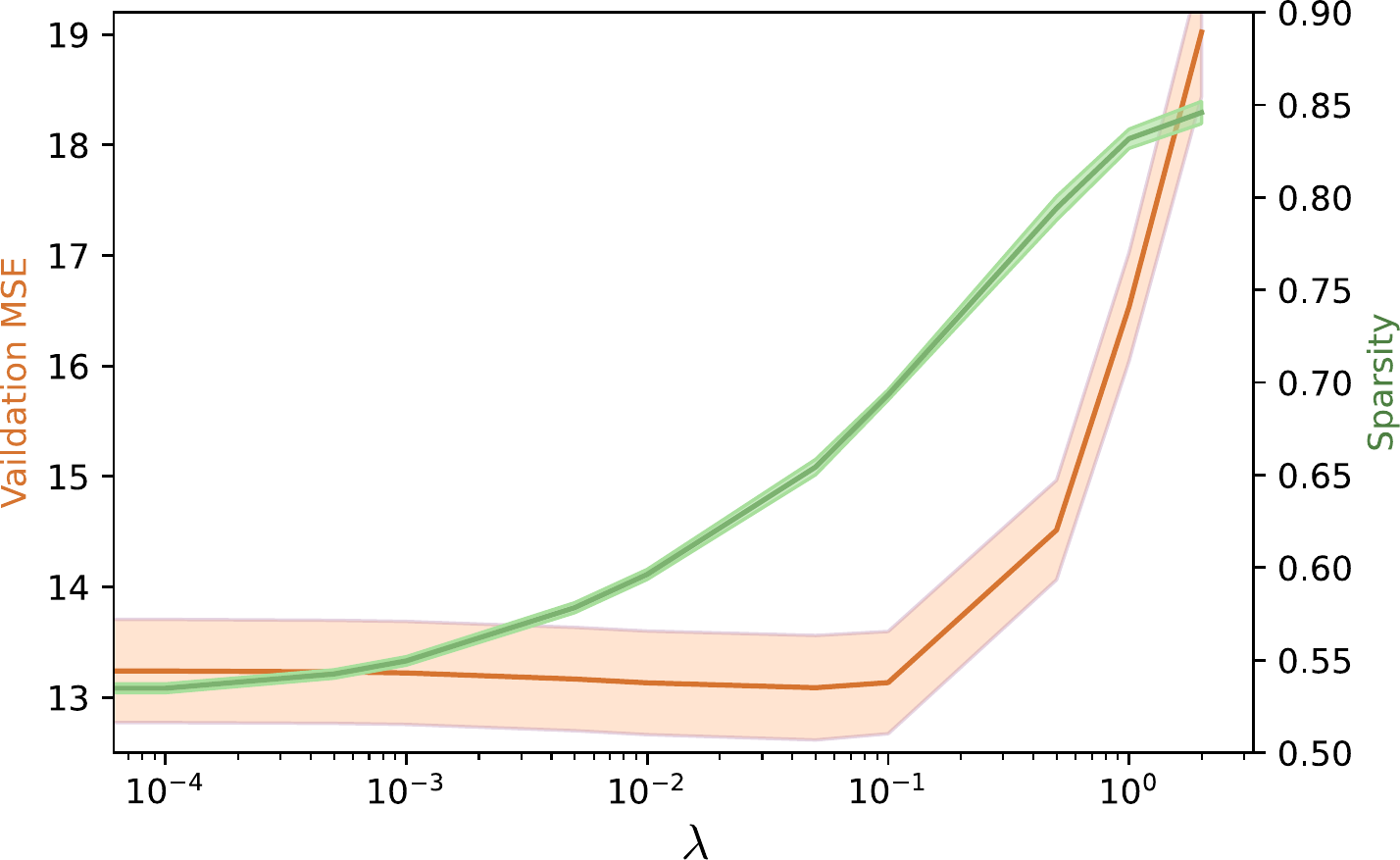}
  \caption{Effect of the regularization hyperparameter $\lambda$ on the validation mean-square error and the sparsity metric throughout the DHTV learning process of the 4-dimensional data. The highlighted area corresponds to the $95\%$ confidence interval.}
  \label{fig:delcpwl}
\end{center}
\end{figure}
\subsection{Combined-Cycle Power-Plant Data Set}
For the experiment of this section, we use a data set that forms a 4-dimensional mapping $f: \R^4 \rightarrow \R$. The data contain 9568 samples. They record the net hourly electrical-energy output of a power plant in terms of temperature, ambient pressure, relative humidity, and exhaust vacuum \cite{tufekci2014prediction}. We compare four learning schemes: linear regression (LR), our framework (DHTV), neural networks (NN), and a kernel method with Gaussian radial-basis functions (RBF). For all learning schemes, we perform 30 independent runs. At each run, we split the data randomly into train ($70\%$), validation ($15\%$), and test ($15\%$) points. In our framework, at each run we perform a grid search on the values of the regularization hyperparameter $\lambda$ and retain the best $\lambda$ in terms of validation error. For the neural network, we try three different fully connected networks: the first is a two-layer neural network (NN2) with 500 units in the hidden layer; the second is a six-layer network (NN6) with 200 units per hidden layer; and the third is a thirty-six-layer network (NN36) with 200 hidden units per layer, where the purpose is to demonstrate overfitting. To train the networks, we use the ADAM optimizer \cite{kingma2014adam}, a batch size of 32, and 10000 epochs. We set the learning rate to $0.001$ at the beginning and reduce it by a factor of 10 in the 400th, 600th, and 800th epoch. For the first two architectures (NN2 and NN6), we perform at each run a grid search on the weight-decay hyperparameter. For the last network (NN36), the weight decay is set to zero. We also perform a grid search on the hyperparameters of the RBF method and choose the model with the best validation error. To achieve a fair comparison of the $\HTV$ estimations of each model, we construct a random triangulation with 1000 grid points. The coordinates of these grid points follow a standard normal distribution in each dimension. We sample the final mapping of each model on the random grid points and calculate the $\HTV$ of the mapping using the sampled function values $\V c_R$, the regularization operator of the random triangulation $\M L_R$, and Formula (40). In addition, we define a metric of sparsity as $\frac{\norm{\abs{\M L_R \V c_R } \leq \epsilon}_0}{\text{Number of rows of }\M L_R} \times 100$. It corresponds to the percentage of almost-coplanar linear pieces over neighboring pairs of simplices. In practice, we set $\epsilon = 0.1$. We report in Table \ref{table:4dresults} the metrics averaged over 30 random splittings. The reported $\HTV$ is normalized by the mean $\HTV$ of the interpolation result or, equivalently, the generated mapping when $\lambda = 0$ in the DHTV framework. We observe that DHTV performs well in terms of prediction error in comparison with other methods. Our parameterization results in a mapping with low fitting error and good prediction power while having fewer learnable parameters than neural networks. An important feature of our method is the control of the complexity of the model by the value of $\lambda$. This is illustrated in Figure 5. The DHTV framework yields sparser mappings while maintaining or even increasing the generalization performance. In addition, as we can see in Table \ref{table:4dresults}, the $\HTV$ and sparsity metrics are in correspondence with the prediction error of the various models, confirming that $\HTV$ is a good metric for the model complexity.

\section{Conclusion}
We have proposed a novel regression method referred to as Delaunay Hessian total variation (DHTV). Our approach provides continuous and piecewise-linear (CPWL) mappings---the same class of functions generated by ReLU networks. We employ a parameterization based on the Delaunay triangulation of the input space. This parameterization represents any CPWL function by its samples on a grid. Unlike deep networks, it is straightforward to understand the effect of the model parameters on the mapping, which makes the proposed model directly interpretable. We have formulated the learning process as a convex minimization task. The Hessian total-variation ($\HTV$) regularization was used to control the complexity of the generated mapping. $\HTV$ has an intuitive formula for CPWL functions which involves a summation over individual affine pieces. We have showed that the $\HTV$ of the proposed model is the $\ell_1$-norm applied to a linear transformation of the grid-point values. This result has enabled us to recast the learning problem as the generalized least absolute shrinkage and selection-operator. By a clever choice of the grid points, we use the fast iterative shrinkage-thresholding algorithm to solve the optimization problem. Our experiments show that the $\HTV$ regularizing leads to simple models while preserving the generalization power. In future works, we plan to investigate the removal of unnecessary grid points based on their contribution to the $\HTV$, which could be used for mesh-simplification purposes.

\section{Acknowledgement}
The authors would like to thank Shayan Aziznejad and Joaquim Campos for having fruitful discussions. 

\bibliographystyle{IEEEtran}
\bibliography{IEEEabrv,dhtv.bib}

\end{document}